\documentclass{article}
\usepackage[margin=1in]{geometry}

\usepackage{moreverb,url}

\usepackage[colorlinks,bookmarksopen,bookmarksnumbered,citecolor=red,urlcolor=red]{hyperref}

\usepackage{amsfonts,amsmath,amssymb}
\usepackage{graphicx}
\usepackage{booktabs}
\usepackage{url}
\usepackage[superscript,biblabel]{cite}

\newcommand{\E}{\mathbb{E}}

\newcommand{\PR}{\mathbb{P}} 
\newcommand{\A}{\mathcal{A}}
\usepackage{tikz}
\usetikzlibrary{shapes,decorations,arrows,calc,arrows.meta,fit,positioning}
\tikzset{
    -Latex,auto,node distance =1 cm and 1 cm,semithick,
    state/.style ={circle, draw, minimum width = 0.5 cm},
    point/.style = {circle, draw, inner sep=0.04cm,fill,node contents={}},
    bidirected/.style={Latex-Latex,dashed},
    el/.style = {inner sep=2pt, align=left, sloped}
}\usepackage{scalefnt}
\usepackage{float}
\usepackage{trivfloat}
\trivfloat{List}
\newcommand{\indep}{\perp \!\!\! \perp}

\usepackage{amsthm}
\usepackage{tabulary}
\usepackage{bbm}

\newcommand{\R}{\mathbb{R}}

\theoremstyle{definition}

\newtheorem{lemma}{Lemma}[section]
\newtheorem{theorem}{Theorem}[section]

\newtheorem{observation}{Observation}[section]
\newtheorem{assm}{Assumption}[section]
\newtheorem{corr}{Corollary}[section]

\begin{document}

\title{
A proxy-based approach for unmeasured confounding in electronic health records research
}

\author
{Haley Colgate Kottler haley.kottler@wisc.edu \\
Department of Mathematics and Department of Emergency Medicine,\\ University of Wisconsin - Madison, Madison, WI, USA \\~\\
and\\~\\
Amy Cochran cochran4@wisc.edu \\
Department of Mathematics, Department of Population Health Sciences,\\ and Department of Emergency Medicine,\\ University of Wisconsin - Madison, Madison, WI, USA}

\maketitle

\begin{abstract}
Electronic health records (EHR) are widely used to study clinical decisions, yet unmeasured confounding remains a persistent challenge. Proxy variables offer a potential solution. In EHR data, clinicians already record many such measurements (e.g., vitals), each revealing something about a patient's underlying health. Despite this, proxy-based methods are rarely used in practice. We introduce a new way to use proxies to adjust for unmeasured confounding. Our approach uses a vector of proxies to construct covariates that capture aspects of the unmeasured confounder, which are then included in a regression model. As one implementation, we use factor analysis followed by regression. We compare this approach with existing methods, including proximal causal inference, across a range of realistic settings. In practice, assumptions rarely hold exactly, so we study what happens when models are misspecified and variables are used incorrectly: e.g., a confounder or instrument is treated as a proxy. Finally, we apply the method to EHR data to estimate the effect of hospital admission for older adults presenting to the emergency department with chest pain, a setting where unmeasured confounding is a substantial concern. This work provides a practical way to use proxies and may help bring proxy-based methods into broader use.
\end{abstract}


\section{Introduction}
Electronic health records (EHR) are a useful resource for improving clinical decision-making, but even detailed datasets like EHR suffer from unmeasured confounding. Natural experiments and instrumental variables (IV) offer solutions, but they rely on conditions that often do not hold in practice. As a result, there is a need for approaches that can be applied more broadly. One promising but underutilized strategy leverages indirect measurements of the unmeasured confounder, called proxies.  Proxies are routinely collected in EHR data (e.g., vital signs may reflect an underlying health condition), and therefore offer a practical opportunity to address unmeasured confounding. Despite this, proxy-based methods remain underused in applied research.

Because proxies resemble other pre-treatment variables, it is tempting to treat them as measured confounders and adjust for them directly. Early work by Greenland argued that adjusting by a binary proxy reduces bias by pulling estimates of an odds ratio toward the truth.\cite{greenland80}  However, Ogburn and Vanderweele later showed that this approach can instead amplify bias when estimating the average treatment effect (ATE).\cite{ogburnVanderweele12} More recent work identifies settings where this approach can help,\cite{Pena21} but relies on restrictive assumptions and data types.

Another idea is to model the full joint distribution of the proxies, treatment, outcome, and unmeasured confounder. In some applications, this has reduced bias and even reversed conclusions in ways that better align with clinical expectations.\cite{Amy,avendano2020average, avendano2022average, avendano2025revisits,gustavson2022handling} More flexible versions use neural networks to learn this joint distribution, but come with limited theoretical guidance on when identification holds.\cite{DNN}  In high stakes fields like healthcare, reliance on a fully specified, yet largely unverified, model of the data-generating process may discourage use.

\begin{figure}
\centering
\scalefont{1}
\tikzset{state/.style={draw,circle}}
\begin{tikzpicture}
\begin{scope}
    \node at (-1.95,1) {\Large \textbf{A)}};
    \node[state] (a) at (0,0) {$A$};
    \node[state] (y) at (2,0) {$Y$};
    \node[state] (u) at (1,1) {$U$};
    \node[state] (z) at (-1,1) {$Z$};
    \draw[very thick] (u.south west) -- (a.north east);
    \draw[very thick] (u.south east) -- (y.north west);
    \draw[very thick] (a.east) -- (y.west);
    \draw[very thick] (u.west)--(z.east);
\end{scope}
\begin{scope}[xshift=150]
    \node at (-1.95,1) {\Large \textbf{B)}};
    \node[state] (a) at (0,0) {$A$};
    \node[state] (y) at (2,0) {$Y$};
    \node[state] (u) at (1,1) {$U$};
    \node[state] (v) at (-1,1) {$V$};
    \node[state] (w) at (3,1) {$W$};
    \draw[very thick] (u.south west) -- (a.north east);
    \draw[very thick] (u.south east) -- (y.north west);
    \draw[very thick] (a.east) -- (y.west);
    \draw[very thick] (u.west)--(v.east);
    \draw[very thick] (v.south east) -- (a.north west);
    \draw[very thick] (u.east)--(w.west);
    \draw[very thick] (w.south west) -- (y.north east);
\end{scope}
\end{tikzpicture}
\caption{DAGs illustrating unmeasured confounding and proxy variables: A) A proxy variable $Z$ influenced by an unmeasured confounder $U$ but independent of both treatment $A$ and outcome $Y$ given $U$. B) Proxy variables $V$ and $W$, respectively called a negative control exposure and negative control outcome, both influenced by an unmeasured confounder $U$ but still dependent on either the treatment $A$ or outcome $Y$ given $U$.}
\label{fig:dag}
\end{figure}

A different line of work classifies proxies based on their causal roles. Figure~\ref{fig:dag} uses directed acyclic graphs (DAGs) to illustrate this idea. Here, $U$ is an unmeasured confounder affecting treatment $A$ and outcome $Y$. A simple proxy, $Z$, is influenced by $U$ but has no direct relationship with $A$ or $Y$. In contrast, a negative control exposure (NCE), $V$, may influence the treatment but not the outcome, while a negative control outcome (NCO), $W$, may influence the outcome but not the treatment.\cite{cui2023semiparametric,TchetgenCOCA} 

Early work used a single proxy to detect unmeasured confounding,\cite{YANG2024111228} with targeted strategies for selection bias and measurement bias.\cite{arnold2016brief} Later work showed how an NCO can correct for unmeasured confounding, starting with the Control Outcome Calibration Approach (COCA) \cite{TchetgenCOCA}  and its extensions.\cite{sofer2016negative, tchetgen2023single}  Using a post-treatment measurement as a proxy for the treatment-free potential outcome, COCA assumes that any confounding between the treatment and the NCO can be accounted for by the observed covariates and the treatment-free potential outcome.\cite{YANG2024111228}  This strategy targets the treatment effects among the treated rather than the ATE.

A major advance came from Kuroki and Pearl\cite{KurokiPearl2014} with subsequent generalization by Miao et al.\cite{Miao} showing that the ATE can be identified with two proxies: an NCE and an NCO, as in Figure~\ref{fig:dag}B.  Their method, termed Proximal Causal Inference (PCI), has become a central approach. In practice, PCI requires solving a \textit{bridge function} equation,\cite{Shi2020Selective} which requires parametric modeling to stabilize estimation since even small errors can lead to ``excessive uncertainty in obtaining a solution to the equation."\cite{tchetgen2020introduction} Liu et al. further note ``these methods can be challenging to implement because they involve solving complex integral equations that are typically ill-posed."\cite{regressionProximal}

Despite strong theory and the widespread availability of proxies in EHR data, PCI remains underused in applied research.  A recent review attributes the gap to ``insufficient guidance on the application in research practice."\cite{YANG2024111228}  Regression-based versions improve accessibility to PCI using generalized linear models.\cite{regressionProximal} These methods accommodate log, logit, and identity link functions for both the outcome and NCO models, treating each combination separately and yielding a set of case-specific procedures. But they impose specific modeling assumptions, including no interaction between $U$ and $A$ and a scalar confounder unless both the outcome and NCO are polytomous. Whether these advances will lead to broader use remains to be seen. 

In this work, we demonstrate that proxy-based adjustment can be carried out through a different route, broadening the set of approaches available in practice. Rather than focusing on NCE–NCO pairs as in PCI, we consider \textit{a vector} of proxies that function as both, as in Figure~\ref{fig:dag}A. The general approach has two steps. First, we use the proxies to construct covariates that summarize aspects of the unmeasured confounder. Second, we include these covariates in an outcome model to adjust for confounding. As one concrete implementation, we use factor analysis to construct these covariates, yielding a simple procedure based on familiar tools: factor analysis followed by regression. We compare this approach side by side with existing methods, including PCI and inverse probability weighting. Because assumptions rarely hold in practice, we focus on how they perform when models are misspecified or when variables are used incorrectly as proxies, confounders, or instruments. Finally, we illustrate its use in an EHR-based case study.

\section{Causal Framework}\label{sec:causal}

\begin{figure}
\centering
\scalefont{1}
\tikzset{state/.style={draw,circle}}
\begin{tikzpicture}
\node at (-1.95,2.5) {\Large \textbf{A)}};
\begin{scope}[yshift = 0]
     \node[state] (x) at (1,2.5) {$X$};
    \node[state] (a) at (0,0) {$A$};
    \node[state] (y) at (2,0) {$Y$};
    \node[state] (u) at (1,1) {$U$};
    \node[state] (z) at (-1,1) {$Z$};
    \draw[very thick] (u.south west) -- (a.north east);
    \draw[very thick] (u.south east) -- (y.north west);
    \draw[very thick] (a.east) -- (y.west);
    \draw[very thick] (u.west)--(z.east);
    \draw[very thick] (x.south) -- (u.north);
    \draw[very thick] (x.south west) -- (a.north);
    \draw[very thick] (x.south east) -- (y.north);
    \draw[very thick] (x.west) -- (z.north east);
\end{scope}
\begin{scope}[xshift=0, yshift=-100]
    \node at (-1.95,2.5) {\Large \textbf{B)}};
    \begin{scope}
    \node[state] (x) at (1,2.5) {$X$};
    \node[state] (a) at (0,0) {$A$};
    \node[state] (y) at (2,0) {$Y$};
    \node[state] (u) at (1,1) {$U$};
    \node[state] (z) at (-1,1) {$Z$};
    \draw[very thick] (u.south west) -- (a.north east);
    \draw[very thick] (u.south east) -- (y.north west);
    \draw[very thick] (a.east) -- (y.west);
    \draw[very thick] (u.west)--(z.east);
    \draw[very thick] (z.south east)--(a.north west);
    \draw[very thick] (x.south) -- (u.north);
    \draw[very thick] (x.south west) -- (a.north);
    \draw[very thick] (x.south east) -- (y.north);
    \draw[very thick] (x.west) -- (z.north east);
    \end{scope}
    \begin{scope}[xshift=125]
    \node[state] (x) at (1,2.5) {$X$};
    \node[state] (a) at (0,0) {$A$};
    \node[state] (y) at (2,0) {$Y$};
    \node[state] (u) at (1,1) {$U$};
    \node[state] (z) at (-1,1) {$Z$};
    \draw[very thick] (u.south west) -- (a.north east);
    \draw[very thick] (u.south east) -- (y.north west);
    \draw[very thick] (a.east) -- (y.west);
    \draw[very thick] (z.south east)--(a.north west);
    \draw[very thick] (x.south) -- (u.north);
    \draw[very thick] (x.south west) -- (a.north);
    \draw[very thick] (x.south east) -- (y.north);
    \draw[very thick] (x.west) -- (z.north east);
    \end{scope}
    \begin{scope}[yshift=-125]
    \node[state] (x) at (1,2.5) {$X$};
    \node[state] (a) at (0,0) {$A$};
    \node[state] (y) at (2,0) {$Y$};
    \node[state] (u) at (1,1) {$U$};
    \node[state] (z) at (-1,1) {$Z$};
    \draw[very thick] (u.south west) -- (a.north east);
    \draw[very thick] (u.south east) -- (y.north west);
    \draw[very thick] (a.east) -- (y.west);
    \draw[very thick] (z.east)--(u.west);
    \draw[very thick] (x.south) -- (u.north);
    \draw[very thick] (x.south west) -- (a.north);
    \draw[very thick] (x.south east) -- (y.north);
    \draw[very thick] (x.west) -- (z.north east);
    \end{scope}
    \begin{scope}[yshift=-125,xshift=125]
    \node[state] (x) at (1,2.5) {$X$};
    \node[state] (a) at (0,0) {$A$};
    \node[state] (y) at (2,0) {$Y$};
    \node[state] (u) at (1,1) {$U$};
    \node[state] (z) at (-1,1) {$Z$};
    \draw[very thick] (u.south west) -- (a.north east);
    \draw[very thick] (u.south east) -- (y.north west);
    \draw[very thick] (a.east) -- (y.west);
    \draw[very thick] (z.east)--(u.west);
    \draw[very thick] (z.south east)--(a.north west);
    \draw[very thick] (x.south) -- (u.north);
    \draw[very thick] (x.south west) -- (a.north);
    \draw[very thick] (x.south east) -- (y.north);
    \draw[very thick] (x.west) -- (z.north east);
    \end{scope}
\end{scope}
\end{tikzpicture}
\caption{DAGs describing unmeasured confounding: A) with a general proxy variable, B) with alternative causal models that satisfy Corollary \ref{corr:proof}.}
\label{fig:dag2}
\end{figure}

We consider the following random variables:  measured confounders $X \in \mathbb{R}^l$; unmeasured confounders $U\in\mathbb{R}^k$; proxy variables $Z\in\mathbb{R}^p$; treatment assignment $A\in\mathcal{A}$, typically with $\A=\R_{\geq0}$ or $\A=\{0,1\}$; and an outcome $Y\in\mathbb{R}$. 
These variables capture events related to a patient's visit to a healthcare provider: Each patient visit has some measured confounders $X$ (like age or sex) along with unmeasured confounders $U$, such as their needs for treatment and monitoring.  The provider collects observations $Z$ of the patient's health, such as their acuity level or imaging results.  Here, $Z$ is assumed to carry imperfect information about $U$ so that $Z$ acts as a proxy for $U$. Importantly, $Z$ contains only the observations that will eventually be available in EHR. After collecting information about the patient, including the information not recorded in $Z$, the provider assigns treatment $A$. Possible treatments include a medication dosage or an admission into the hospital. It is assumed that $U$ can be defined so that the treatment assignment depends on $Z$ only via its relation to the latent variable $U$. After treatment is assigned, the patient is followed up for some time period, and an outcome $Y$ such as an adverse health event is observed. 

We embed these variables into a non-parametric structural equation model (NPSEM) that consists of independent random noises $\{\varepsilon_X,\varepsilon_U,\varepsilon_Z,\varepsilon_A,\varepsilon_Y\}$ and deterministic ``assignment" functions $\{f_X,f_U,f_Z,f_A,f_Y\}$ such that our random variables are assigned iteratively:
\begin{align*}
    X &:= f_X(\varepsilon_X),&&&
    U &:= f_U(X, \varepsilon_U),\\
    Z &:= f_Z(X, U, \varepsilon_Z), &&&
    A &:= f_A(X, U, \varepsilon_A),\\
    Y &:= f_Y(X, A, U, \varepsilon_Y).
\end{align*}
One practical benefit of NPSEMs is their independence assumptions can be represented concisely in a DAG; each node is a variable and a directed edge exists from one variable to another if the former is an argument in the assignment of the latter. This NPSEM results in Figure \ref{fig:dag2}A.  We later relax these structural assumptions.

Our NPSEM defines potential outcomes, following the single world intervention graph (SWIG) framework introduced by Richardson and Robins.\cite{richardson2013single} This framework offers us the versatility of potential outcomes as described by Neyman\cite{neyman1923application} and Rubin,\cite{rubin1974estimating} with the precision and graphical aids of do-calculus described by Pearl et al.\cite{pearl2016causal} 
The idea is to replace the observed intervention $A$ with the fixed intervention $a\in\A$ in all subsequent assignments. This procedure allows us to iteratively define a new set of variables:
\begin{align*}
    X(a) &:= f_X(\varepsilon_X),
    &&& U(a) &:= f_U(X(a),\varepsilon_U),
    \\Z(a) &:= f_Z(X(a),U(a), \varepsilon_Z), &&&
    A(a) &:= f_A(X(a),U(a), \varepsilon_A),\\
    Y(a) &:= f_Y(X(a),a, U(a), \varepsilon_Y).
\end{align*}
The notation ``$(a)$" serves to mark that these new variables result from fixing $a$ may not be the originally observed variables. 
Our assumptions allow us to immediately recover the relations: $X(a)=X$, $U(a)=U$, $Z(a)=Z$, and $A(a)=A$.
In the SWIG framework, these relations follow from the principle of causal irrelevance.\cite{richardson2013single} Accordingly, we drop ``$(a)$" for variables other than $Y(a)$ from subsequent exposition. 

\section{Identification}\label{sec:identification}
With the potential outcomes defined, our goal is to use the observed variables to estimate the conditional average treatment effect (CATE): $$\beta(Z,X):=\mathbb{E}[Y(a_1)-Y(a_0)|Z,X]$$
for $a_1,a_0 \in \cal A$. The CATE tells us the expected difference in outcome if someone with $Z$ and $X$ is given treatment $a_1$ instead of treatment $a_0$.  By conditioning on $Z$ and $X$, we are able to gauge the effect of the intervention on a specific group with $Z$ and $X$, instead of in the population as a whole.  This is particularly useful in applications where treatment effects may vary across subgroups. Estimating the CATE allows for more targeted inference.

Learning about the CATE is challenging because we only observe what occurs when a patient is assigned to one treatment. We cannot directly estimate the CATE without further manipulation.  Identification is the process by which we re-express our causal identity of interest, in our case the CATE $\beta(Z,X)$, in terms of the observed data. 

\subsection{Why direct adjustment fails}

It is useful to appreciate how a strategy that adjusts for proxies can go wrong. For simplicity, let $\A = \{0,1\}$ and omit $X$. Imagine we wrongly assume the treatment assignment $A$ and outcome $Y$ are influenced directly by proxies $Z$ rather than $U$.  We would then identify the CATE with the expression:
$$\E[Y|Z, A=1] - \E[Y|Z,A=0],$$
i.e. subtracting the mean outcome for people who did not receive treatment from the mean outcome for people who did receive treatment with the same characteristics $Z$. Suppose further that $Y$ is generated via a simple additive noise model:
$$f_Y(a,u,\varepsilon_Y) = \varepsilon_Y + g_0(u) + a g_1(u),$$
where $g_0$ and $g_1$ are arbitrary and $\E[\varepsilon_Y] = 0$. The only additional assumption here is the additivity of noise. In this case, $Y(1) - Y(0) = g_1(U)$, so the true CATE is simply $\E[g_1(U)|Z]$.

Now consider what the expression $\E[Y | Z, A=1] - \E[Y | Z, A=0]$ actually estimates under this model. It differs from $\E[g_1(U) | Z]$ and instead equals:
\begin{align*}
\E[g_1(U) | Z, A=1] + \left(\E[g_0(U) | Z, A=1] - \E[g_0(U) | Z, A=0]\right).
\end{align*}
Its bias can be written compactly as:
\begin{align*}
\E[h(U, Z) | Z, A=1] - \E[h(U, Z) | Z, A=0],
\end{align*}
where $h(U, Z) = g_0(U) + \PR(A = 0 | Z) g_1(U)$. In general, this bias does not vanish, except under special conditions like $h(U, Z) \indep A \mid Z$.

To make this more concrete, suppose the outcome model is affine in the unobserved variable $U$. Specifically, let
\begin{align*}
g_0(u) = c_0 + c_1' u; &&& g_1(u) = c_2 + c_3'u.
\end{align*}
so the effect of $U$ on $Y$ enters through linear terms. Under this assumption, the bias that arises from adjusting on $Z$ simplifies to:
$$(c_1 + \PR(A=0|Z) c_3 )'\left(\E[U|Z,A=1]-\E[U|Z,A=0]\right).$$
This expression highlights two key drivers of bias: (1) how much $U$ affects the outcome, and (2) how imbalanced $U$ is between treated and untreated individuals with the same value of $Z$. The bias is zero exactly when $c_1 + \PR(A=0|Z) c_3$ and $\E[U|Z,A=1]-\E[U|Z,A=0]$ are orthogonal---roughly speaking, when the direction in which $U$ affects the outcome is perpendicular to the direction in which $U$ is imbalanced across treatment groups.
In particular, bias is small if either $U$ has little or no effect on the outcome ($c_1$ and $c_3$ are small), or the treatment assignment does not tell you anything more about $U$ beyond what is already captured in $Z$ ($\E[U|Z,A=1]-\E[U|Z,A=0]$ is small). Put simply, the bias will be small when the amount of confounding between $A$ and $Y$ that comes from $U$ is small, either because $U$ has little effect on the outcome, or its influence is well-approximated by $Z$.

\subsection{A general strategy for identification}

We have seen that simply adjusting for the proxy variable $Z$ does not recover the CATE, so we need an alternative identification strategy. We begin by outlining the most general version of our strategy, before turning to the practical implementation we use. We rely on two assumptions. First, although $U$ is unobserved, the data still contain information about it. We assume this information is enough to recover certain functions of $U$ in expectation:

\begin{assm} \label{assm:identify}
    There is a vector-valued function $g(U,X)$ so that $\E[g(U,X)|Z,A,X]$ can be identified from the observed variables $(Z,A,X)$.
\end{assm}

This assumption is intentionally flexible. We do not need to recover $U$ itself, just functions of $U$ in expectation. Here, \(\E[g(U,X)\mid Z,A,X]\) plays this role: it depends only on observed variables, but still reflects confounding. In practice, it means we must relate the proxies $Z$ to the latent variable $U$. Many standard models do exactly this by treating $Z$ as noisy measurements of some underlying latent structure. Under this view, methods such as Gaussian mixture models, latent class analysis, or factor analysis provide a way to recover functions of $U$ from $Z$. In Section~\ref{sub:parametric}, we use factor analysis, a standard tool in medical research with a well-developed theoretical background.\cite{classical_latent}

Second, we require these functions are expressive enough to capture how the expected potential outcomes vary with treatment and covariates.  This is essentially a working model for the potential outcome.  That is:
\begin{assm}\label{assm:express}
    The causal quantity $\mathbb{E}[Y(a)|U,X]$ can be written as $\tau(a,X)' g(U,X)$ for some vector-valued function $\tau(a,X)$ and each $a \in \A$.
\end{assm}
This can be viewed as a separability condition on the functional form of the causal quantity $\E[Y(a)|U,X]$.
Combined, these assumptions provide a general strategy for identification: 
\begin{theorem}\label{thm:identify}
    Assume Assumptions~\ref{assm:identify}--\ref{assm:express} hold with vector-valued functions $\tau$ and $g$ defined therein. Then as long as $\tau(A,X)$ is the unique solution to:
    $$\mathbb{E}[Y|Z,A,X]=\tau(A,X)'\mathbb{E}[g(U,X)|Z,A,X],$$ 
    we can identify the CATE by way of
    $$\beta(Z,X)=(\tau(a_1,X)-\tau(a_0,X))'\mathbb{E}[g(U,X)|Z,X].$$
\end{theorem}
We include the proof in Appendix \ref{ap:idproofs}. This theorem outlines a two-step strategy for identifying the CATE. The first step is to target a lower-dimensional summary of $U$, captured by the conditional expectation $\E[g(U, X) | Z, A, X]$. This quantity is both a function of the observed variables and can be recovered from them. Once we have this quantity, we treat it as a predictor and use it to model the conditional mean of the outcome $Y$ by learning the function $\tau(A, X)$ in the expression $\E[Y | Z, A, X] = \tau(A, X)'\E[g(U, X) | Z, A, X].$ This is done using regression. With $\tau(A, X)$ recovered, we then compute the CATE.

We can relax the structural assumptions, allowing for greater flexibility in underlying models, including the possibility of $Z$ directly influencing $A$ or $A$ directly influencing $Z$:

\begin{corr}\label{corr:proof}
    Suppose Assumptions~\ref{assm:identify}–\ref{assm:express} hold, and the NPSEM corresponds to one of the DAGs in Figure~\ref{fig:dag2}B, with potential outcomes defined accordingly. Then the conclusions of Theorem~\ref{thm:identify} hold.
\end{corr}
Each NPSEM in Figure~\ref{fig:dag2}B, together with its associated potential outcomes, implies the independence relations: $Y(a) \indep Z \mid U, X$ and $Y(a) \indep A \mid U, X$. These independence conditions are used to establish conclusions analogous to those in Theorem~\ref{thm:identify} across all the depicted causal structures. This shows that our strategy can tolerate ambiguity in structural assumptions, which is often unavoidable in practice. Even if we cannot determine which NPSEM generated the data, identification still holds as long as one of the models in Figure~\ref{fig:dag2}B is valid. This is especially relevant in applied settings, where a variable assumed to be a proxy might instead be a direct cause of $U$, or an instrument. Our strategy accommodates these distinctions without requiring precise knowledge of each variable’s structural role, making it useful when such details are difficult to verify. Additionally, notice that this allows most pretreatment measures, which are widely available, to act as proxies.

\subsection{A specific, parametric strategy for identification}\label{sub:parametric}

Our ability to gather information about the latent variable (Assumption~\ref{assm:identify}) forms the cornerstone of our identification strategy. Here, we outline how to identify expectations of the form $\E[g(U,X)|Z,A,X]$ by imposing stronger assumptions on the variables $U$, $Z$, $A$, and $X$. Notably, the outcome $Y$ is not included in these assumptions. Our task later will be to show that our final estimate of the CATE is relatively insensitive to these additional assumptions. We consider the multiple indicators multiple causes (MIMIC) model \cite{tekwe2014multiple, chang2020comparing}, an extension of factor analysis that incorporates covariates:

\begin{assm}\label{assm:parametric}
Assume the following parametric model: First, $\varepsilon_X$, $\varepsilon_U$, $\varepsilon_Z$, and $\varepsilon_A$ are all independent multivariate standard normals. Second, assignment functions are given by:
\begin{align*}
f_X(\varepsilon_X) &=\varepsilon_X \Theta^{1/2}
&&& f_U(\varepsilon_U) &= \Gamma X + \varepsilon_U, \\
f_Z(U,\varepsilon_Z) &= \Lambda U + \nu + \Psi^{1/2} \varepsilon_Z;
&&& f_A(U,\varepsilon_A) &= b'U + c + \varepsilon_A,
\end{align*}
where $\Psi \in \mathbb{R}^{p \times p}$ is diagonal and positive definite; $\Lambda \in \mathbb{R}^{p \times k}$ and $\Gamma \in\mathbb{R}^{k\times m}$ are full rank; and $\nu \in \mathbb{R}^p$, $c \in \mathbb{R}$, and $b \in \mathbb{R}^k$, with $k<p$.
\end{assm}

As discussed in Tekwe et al.\cite{tekwe2014multiple} and Chang et al.,\cite{chang2020comparing} parametric models of this form identify the parameters only up to an orthogonal transformation of the latent factor $U$. In applied settings, this indeterminacy is typically resolved by imposing a constraint or applying a rotation post hoc to select among models with equivalent likelihoods. For simplicity, we adopt a standard identifying constraint from the factor analysis literature (though other constraints would serve equally well in what follows):

\begin{assm}\label{assm:posdef}
    Loading matrix $\Lambda$ is lower triangular with positive diagonal entries.
\end{assm}

Under this constraint, we now state an identifiability result:

\begin{lemma} \label{lem:par_ident}
Given Assumptions~\ref{assm:parametric} and~\ref{assm:posdef}, the distribution of $U|Z,A,X$ can be identified.\cite{chang2020comparing}  Therefore, $\E[g(U,X) | Z, A, X]$ can be identified from the observed variables $Z$, $A$, and $X$ for any vector-valued function $g$.
\end{lemma}

This lemma implies that Assumption~\ref{assm:identify} is satisfied. Note also that Assumption~\ref{assm:express} is trivially satisfied, since this lemma puts no restrictions on what functions $g$ can used to capture $\E[Y(a)|U,X]$. Therefore, the CATE can be identified as described above:

\begin{corr}\label{corr:rotation}
     Suppose the potential outcomes are defined, and 
 Assumptions~\ref{assm:parametric} and \ref{assm:posdef} hold.
    We can identify the CATE using the procedure outlined in Theorem~\ref{thm:identify}.
\end{corr}

\section{Estimation}\label{sec:estimation}

We use our identification strategy to guide the steps of estimation in practice. We walk through a concrete example that reflects how an applied investigator would typically approach this problem in practice, and later describe how to estimate effects in other settings.

The applied investigator would first identify the roles of the variables at hand: the treatment $A$, outcome $Y$, measured confounders or covariates $X$, proxy variables $Z$, and possible unmeasured confounders $U$. The investigator would then specify a working model for the potential outcome in terms of unmeasured confounders $U$, treatment $A$, and covariates $X$. For simplicity, we choose the following working model:
\begin{align*}
    \E[Y(a)|X,U] = \nu_0 + \nu_1'U + \nu_2'UA + \nu_3 A. 
\end{align*}  This model can be factored as $(\alpha+\gamma A)'\left[ 1 \quad U \right]'$ for an appropriate choice of $\alpha$ and $\gamma$.  This allows us to see that this model corresponds to $g(U,X) = \left[ 1 \quad U \right]'$ and $\tau(a, X) = \alpha' + \gamma' A$ in the earlier identification proof. We also assume a MIMIC model for $(U,Z,A,X)$ (Assumption \ref{assm:parametric}).
Importantly, our choice of working model allows for closed-form expressions of the conditional means required in the outcome model, enabling efficient computation. When such expressions are not available (e.g., under more complex or nonlinear models), these quantities can be computed using Monte Carlo integration.

Once the working model is specified, we fit the MIMIC model to our observations of $X$, $A$, and $Z$.  Existing latent modeling packages like \textit{lavaan} in R or MPlus can fit the MIMIC model, making this step accessible to applied researchers. Since this requires selecting the number of factors, we fit models with 1 to $p$ factors and select via the lowest Akaike information criterion (AIC) value.
Subsequently, we use the fitted model to estimate two quantities, $\E[U|Z,A,X]$ and $\E[U|Z,X]$, for each patient at their observed values of $Z$, $A$, and $X$. 
Let $\hat\E[U|Z,A,X]$ and $\hat\E[U|Z,X]$ denoted the estimated values. See Supporting~Information~A for the derivation of a closed form for $\E[U|Z,A,X]$ and $\E[U|Z,X]$.  

We then estimate the parameters of the potential outcome model, $\nu_0$, $\nu_1$, $\nu_2$, and $\nu_3$, by fitting a linear regression of the outcome $Y$ on the observed variables $Z$, $A$, and $X$. Under our assumptions, this model takes the form
$$\E[Y|Z,A,X] = \nu_0 + \nu_1'\E[U|Z,A,X]+\nu_2'\E[U|Z,A,X]A + \nu_3A,$$
where the regression coefficients correspond directly to those in the potential outcome model. Since $U$ is unobserved, we substitute in its estimate, $\hat\E[U | Z, A, X]$, obtained from the fitted MIMIC model. The resulting model includes two main effects and one interaction.

Upon performing linear regression, we recover 
estimates of the parameters of the potential outcome model, which we denote by $\hat\nu_0$, $\hat\nu_1$, $\hat\nu_2$, and $\hat\nu_3$. We put together these estimates and our estimate of $\hat\E[U|Z,X]$ to arrive at the CATE estimate:
\begin{align*}
\hat{\beta}(X,Z) = &(\hat\nu_0+\hat\nu_1'\hat\E[U|Z,X]+\hat\nu_2'\hat\E[U|Z,X]a_1 + \hat\nu_3a_1)\\ 
&\quad -(\hat\nu_0+\hat\nu_1'\E[U|Z,X]+\hat\nu_2'\hat\E[U|Z,X]a_0 + \nu_3a_0 ) \\
= &(a_1-a_0)(\hat\nu_2'\hat\E[U|Z,X] + \hat\nu_3).
\end{align*}
The entire estimation procedure can be formalized using estimating equations, with consistency and asymptotic normality provable under standard conditions. However, we omit these derivations here, as they are unlikely to be of practical concern for applied investigators.

\section{Simulation}\label{sec:synthetic}

We evaluate our method using simulated data where the true treatment effect is known. We examine bias, variance, and asymptotic behavior. When the latent variable model is correctly specified, we also consider the ratio of proxies to latent confounders ($p/k$). To assess robustness to model misspecification, we consider scenarios where the mean function is quadratic in $U$ but modeled as linear, where $U$ follows a skew normal rather than a normal distribution, and where the treatment $A$ is binary rather than continuous. Finally, we explore robustness to violations of structural assumptions, including cases where $Z$ is a direct confounder or an IV, but is treated as a proxy. All computations  are completed using the parametric strategy from Section \ref{sec:estimation}, with $g(U,X)=[1,U]'$.

We compare our method to inverse probability weighting (IPW), linear models, an IV approach, and PCI. See Supporting~Information~B for details. Briefly, IPW uses a linear regression model for $A$ in terms of $Z$ to determine the weights and then models $Y$ with a linear model in $A$.\cite{naimi2014constructing}  For binary treatments, the linear model for A is replaced with a logistic regression. The linear method uses a linear regression model for $Y$ in terms of $A$ and $Z$. For the proximal method, we randomly split the proxy variables evenly into NCOs and NCEs. Since the proxies are interchangeable, this is equivalent to averaging over all possible even divisions into NCOs and NCEs. Since most of these techniques estimate the Average Treatment Effect (ATE) $\E[Y(1)-Y(0)]$,  rather than the CATE, we take the expectation of our CATE estimations since $\E[\E[Y(1)-Y(0)|Z]]=\E[Y(1)-Y(0)]$. 

We use the data-generating process described in Assumption~\ref{assm:parametric} as the baseline. The outcome is generated as $Y = (\alpha + A\gamma)[1 \quad U']'$, where $\alpha$ and $\gamma$ are selected to yield a true average treatment effect (ATE) of 0.3. Measured confounders $X$ are ignored. For each setting, we generate 1000 independent samples and apply each method to estimate the treatment effect. This process is repeated across 1000 simulation replicates. Then, for each robustness check, we slightly modify this process in targeted ways to test how well the method performs under deviations from the ideal case.  We follow Tucker et al.\cite{Tucker_Koopman_Linn_1969} and Nieser and Cochran\cite{nieser2023addressing} to generate realistic loading and diagonal covariance matrices with varying levels of communality. We also vary the sample size, the number of proxies ($p$), the number of latent dimensions ($k$), and the strength of confounding. To evaluate our method's coverage, we construct 95\% bootstrap confidence intervals. Source code can be found at \url{https://github.com/HaleyColgateKottler/LatentProxyVars},
and details in Supporting~Information~B. Since the IV method is sensitive to assumption violations, it is excluded from the main plots for readability. For results including IV as a comparator, see Supporting~Information~B.

Figure \ref{fig:robustpanel1}A illustrates estimation accuracy across varying sample sizes.  As the sample size increases, the 25th and 75th percentiles for our method's estimates tighten around the true effect value of $0.30$, suggesting consistency.  At smaller sample sizes, our estimates show greater accuracy than the comparison techniques.

\begin{figure}
    \centering
    \includegraphics[width=.9\linewidth]{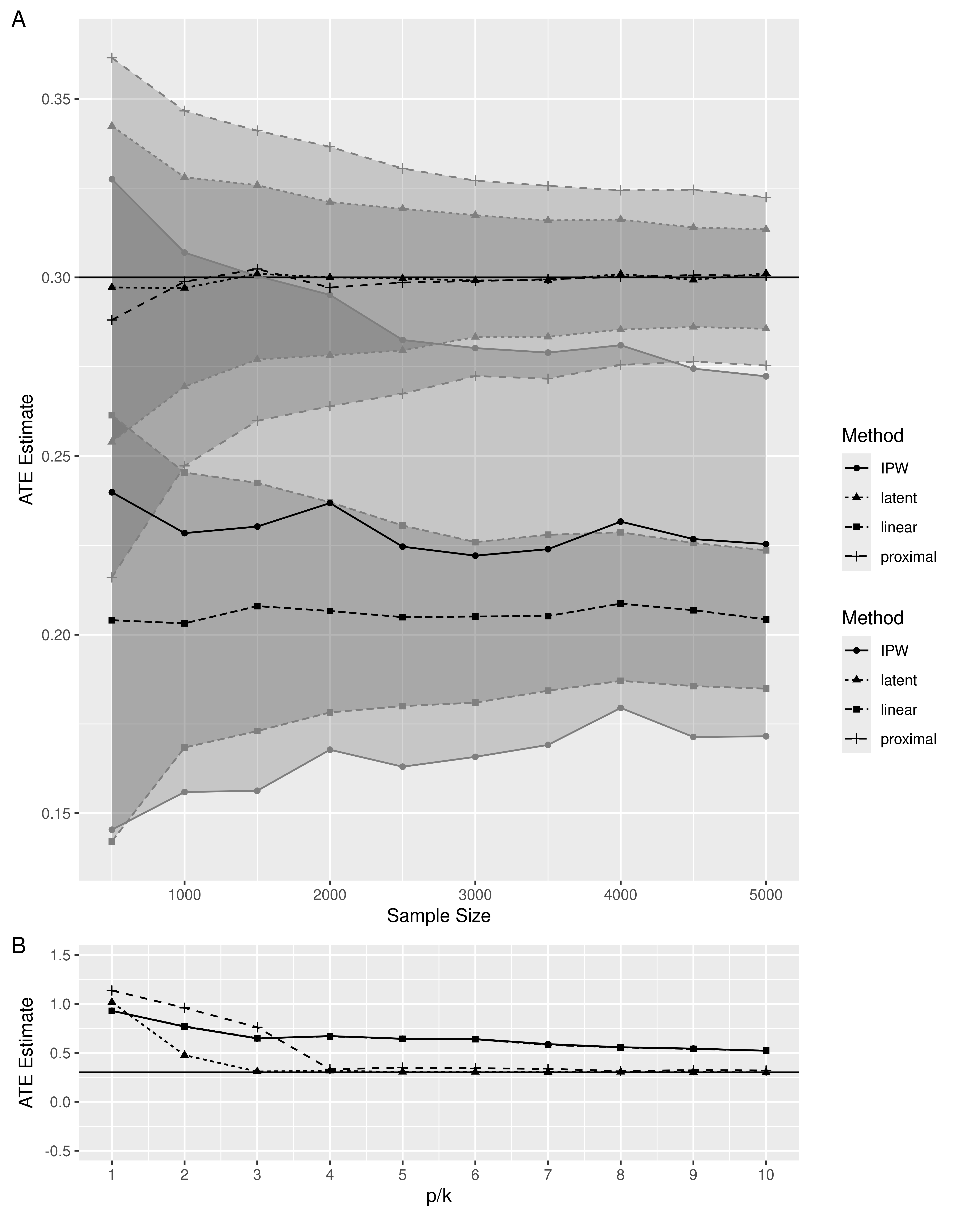}
    \caption{A) Estimation accuracy of our method and comparison methods as a function of sample size, reported as the median estimate and interquartile ranges (25th and 75th percentiles) across simulation replications. B) Estimation accuracy as function of different ratios of proxies to latent confounders ($p/k$), reported as the median estimate across simulation replications. Our method is the most accurate for all sample sizes, rivaled only by proximal causal inference which has a higher mean squared error.  With $p/k\geq2$ our method is more accurate than the comparison methods, though with higher $p/k$ ratios, proximal causal inference obtains similar accuracy.}
    \label{fig:robustpanel1}
\end{figure}

Figure \ref{fig:robustpanel1}B shows our method is more accurate with larger ratios of proxies to latent confounders ($p/k$). Accuracy improves notably when this ratio increases from 1 to 2 and from 2 to 3, indicating that having three proxies per latent dimension is particularly beneficial. This aligns with expectations, as factor analysis models typically require at least three observed variables per latent factor. With at least three proxies per dimension, our method outperforms all comparison methods in terms of how close the median is to the true effect, until the PCI method catches up around a $p/k$ ratio of 8. 

For a sample size of 500, bootstrap confidence intervals give a coverage of 96\%, using 2000 resamples and 150 simulation replications. Since bootstrapping accounts for variability from both modeling steps, and is generally conservative, this is a reasonable coverage estimate.  

\begin{figure}
    \centering
    \includegraphics[width=\linewidth]{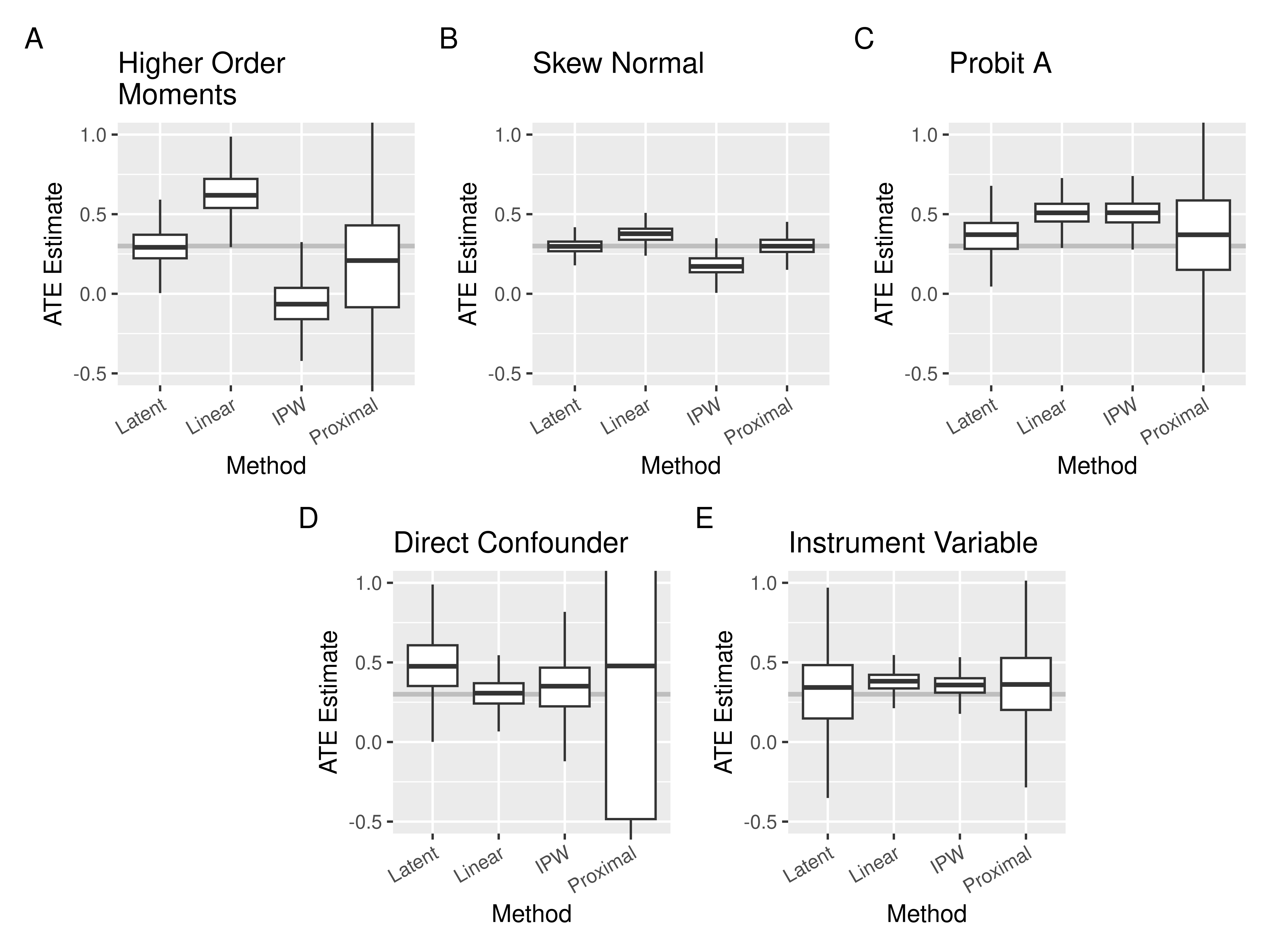}
    \caption{Performance of methods under different modeling and structural assumption violations: A) the true outcome generation process is quadratic in $U$ but the estimator is linear, B) the latent variable is skew normal instead of normal, C) the treatment is binary but estimated as continuous, D) $Z$ is a direct confounder not a proxy, E) $Z$ is an instrument not a proxy. The centerline is the median, the hinges show the inter-quartile range (IQR), and the whiskers show the minimum of the highest estimate or 1.5IQR from the upper hinge and the maximum of the lowest estimate or 1.5IQR from the lower hinge. }
    \label{fig:robustpanel2}
\end{figure}

Figure \ref{fig:robustpanel2}A demonstrates the ability of our method to maintain accurate estimates despite an incorrect mean-model. In this scenario, we assume for estimation that $\E[Y(a)|U]=\left[1 \quad U\right]\tau(a,X)$, which can be seen as a first moment approximation, but generate data with $\E[Y(a)|U]=\left[ 1 \quad U \quad U^2 \right]\tau(a,X)$ instead.  We see that with this adjustment in data generation, all the comparison methods are biased, while our estimate remains essentially unbiased.

Our approach can be reasonable even when the distribution of the latent variable is misspecified (Figure \ref{fig:robustpanel2}B). In this scenario, the data is generated with a skew normal distribution for $U$, even though our method assumes a normal distribution. Our method exhibits lower bias than all comparison methods except proximal causal inference.

Figure \ref{fig:robustpanel2}C illustrates our method when applied to binary treatments. We generate a binary intervention using a probit model, $A^* = \mathbb{I}_{A \geq 0}$, but continue to assume a continuous treatment during estimation. All methods exhibit some bias under this misspecification, though our method and PCI are notably less biased than the others. However, PCI produces estimates with higher variance, resulting in lower overall accuracy.

We then consider violations in structural assumptions (Figure \ref{fig:robustpanel2}D). In Section \ref{sec:identification}, we showed how treating a proxy as a direct confounder introduces bias.  Here we consider the opposite: $Z$ is a direct confounder, not a proxy. As expected, linear and IPW techniques---which correctly treat $Z$ as a confounder---are less biased than ours. Nevertheless, our method still recovers the correct sign and order of magnitude. So even when a confounder is incorrectly assumed to be a proxy variable, our method can produce reasonable results.

Figure \ref{fig:robustpanel2}E explores the impact of treating an IV as a proxy. We apply our technique to data where $Z$ and $U$ are independent multivariate normal random variables, and generate $A$ via $A=b_1'U+b_2'Z + c + \varepsilon_A$, with $\varepsilon_A\sim \mathcal{N}(0,1)$. In this setting, our method exhibits a negative bias, while the comparison methods are biased in the opposite direction. IV is the only method that yields an unbiased estimate (Supporting~Information~B), as expected.  

It is important to note that in the structural misspecification tests, our method produced some estimates with absolute error of at least 1 (25/1000 trials for the direct confounder case, 71/1000 trials for the IV case). This occurs primarily because $Z_i\indep Z_j$ for all $i\neq j\in1,\dots,p$ and $A$ is essentially a linear combination of entries in $Z$ and independent noise, since $U\indep Z$.  Factor analysis then often recovers a single factor model where the largest loading is on $A$, with minimal loadings on the components of $Z$.  This creates colinearity in the regression of $Y$ onto $A$ and $U$, creating outliers.  In usual practice, this would be noticed, but it remains a problem in simulations with many trials.

For the cases when the model is incorrect, we computed our method's coverage using 100 replications with 500 bootstrap replications. We achieve a coverage of 0.96 for the higher order outcome generation, 0.94 for the skew normal latent variable, 0.87 for the binary treatment, 0.97 for the direct confounder treated as a proxy, and 0.94 for the IV treated as a proxy. The coverage is slightly lower in the binary case, but otherwise similar to the coverage when the model is correctly specified.

\section{Case Study}\label{sec:case}

We applied our method to investigate the impact of hospitalization or discharge on older adults in the ED with chest pain. Older adults (65+ years) frequently present to EDs with chest pain, a nonspecific symptom that complicates decisions regarding whether they should be admitted or discharged.\cite{avendano2020average} This decision is clinically significant because hospital admission poses risks for older adults, such as deconditioning and hospital-acquired infections, but chest pain can also be indicative of serious conditions so not admitting a patient also poses risks.

Our analysis uses EHR data from MIMIC-IV v2.2, a large, de-identified database containing real hospital records from a tertiary medical center in Boston, Massachusetts.\cite{johnson2020mimic}  We restricted the dataset to patients aged 65 or older who presented to the ED with “chest pain” listed in their chief complaint and were assigned an acuity level of 2 or 3 since those are the patients for whom the decision to admit or discharge is unclear. We excluded records where the disposition was neither admission nor discharge. Table \ref{tab:1} provides an overview of the data, including specific lists of covariates and proxies.  Disposition is our treatment, and readmission or mortality within 30 days of discharge from any unit (ED or inpatient) is our outcome of interest.  We treat demographic data as covariates.  Data gathered during triage, along with total treatment time, are treated as proxies.

\begin{table}
    \centering
    \begin{tabular}{lllll}
    \toprule
         & & \bf{Admitted} & \bf{Home} & \bf{Overall} \\ \midrule
         \textbf{Covariates}
         & \textbf{N (\% present)} & \textbf{Mean (SD)} & \textbf{Mean (SD)} & \textbf{Mean (SD)}\\ \midrule
        Age & 7081 (100) & 75.4 (7.5) & 73.8 (7.1) & 74.7 (7.4)\\\midrule
        & \textbf{Level} & \textbf{Count (\%)} & \textbf{Count (\%)} & \textbf{Count (\%)}\\\midrule
        Gender & Female & 1930 (48.0) & 1798 (58.8) & 3728 (52.6)\\
        & Male & 2093 (52.0) & 1260 (41.2) & 3353 (47.4)\\
        Race/Ethnicity
        & American Indian/\\
        &Alaska Native & 6 (0.1) & 5 (0.2) & 11 (0.2) \\
        & Asian & 155 (3.9) & 124 (4.1) & 279 (3.9)\\
        & Black & 632 (15.7) & 623 (20.4) & 1255 (17.7)\\
        & Hispanic/Latinx & 227 (5.6) & 249 (8.1) & 476 (6.7)\\
        & Other & 181 (4.5) & 172 (5.6) & 353 (5.0)\\
        & White & 2787 (69.3) & 1881 (61.5) & 4668 (65.9)\\
        & Missing & 35 (0.9) & 4 (0.1) & 39 (0.6)\\
        Insurance
        & Medicaid & 100 (2.5) & 42 (1.4) & 142 (2.0)\\
        & Medicare & 2545 (63.3) & 918 (30.0) & 3463 (48.9)\\
        & Other & 1341 (33.3) & 517 (16.9) & 1858 (26.2)\\
        & Missing & 37 (0.9) & 1581 (51.7) & 1618 (22.8)\\
        Marital status
        & Divorced & 347 (8.6) & 127 (4.2) & 474 (6.7)\\
        & Married & 1846 (45.9) & 687 (22.5) & 2533 (35.8)\\
        & Single & 684 (17.0) & 323 (10.6) & 1007 (14.2)\\
        & Widowed & 1073 (26.7) & 332 (10.9) & 1405 (19.8)\\
        & Missing & 73 (1.8) & 1589 (52.0) & 1662 (23.5)\\\midrule
        \textbf{Proxies} & \textbf{N (\% present)} & \textbf{Mean (SD)} & \textbf{Mean (SD)} & \textbf{Mean (SD)}\\ \midrule
        Heart rate & 7073 (99.9) & 79.2 (17.4) & 76.3 (14.7) & 78.0 (16.4)\\
        Systolic bp& 7069 (99.8) & 140.6 (25.4) & 146.8 (23.1) & 143.3 (24.6)\\
        Diastolic bp&7060 (99.7) & 73.5 (15.3) & 75.7 (14.3) & 74.5 (14.9)\\
        Respiration rate & 7014 (99.1) & 18.1 (2.8) & 17.6 (1.9) & 17.9 (2.5)\\
        O2 saturation & 7038 (99.4) & 97.6 (3.1) & 98.3 (1.7) & 97.9 (2.6)\\
        Temperature & 6971 (98.4) & 98.0 (1.0) & 97.9 (0.8) & 98.0 (0.9)\\
        Treatment time & 7148 (100) & 7.1 (4.5) & 13.9 (9.0) & 10.1 (7.6)\\\midrule
        & \textbf{Level} & \textbf{Count (\%)} & \textbf{Count (\%)} & \textbf{Count (\%)}\\\midrule
        Acuity & 2 & 3469 (86.2) & 2181 (71.3) & 5650 (79.8)\\
        & 3 & 554 (13.8) & 877 (28.7) & 1431 (20.2)\\\midrule
        \textbf{Outcome} & \textbf{Level} & \textbf{Count (\%)} & \textbf{Count (\%)} & \textbf{Count (\%)}\\\midrule
        30 day & Yes & 888 (22.1) & 329 (10.8) & 1217 (17.2) \\
        readmission & No & 3135 (77.9) & 2729 (89.2) & 5864 (82.8)\\\bottomrule
    \end{tabular}
    \caption{Demographics, triage data, and outcome by disposition\\bp: blood pressure}
    \label{tab:1}
\end{table}

We handled missing data in several ways.  First, we excluded any records missing the treatment or outcome, or missing three or more variables out of the proxies or covariates.  We introduced an ``Unknown" category for insurance type, marital status, and race to account for possible differences in patients without full demographic data. Finally, we applied multiple imputation once using the MICE R package\cite{mice} including all variables.  This resulted in a final dataset of 7081 ED visits. For categorical variables, we treat the most common level as the reference class.  To incorporate covariates, we adopted the MIMIC model as described in Section \ref{sub:parametric}.  For clarity, we use MIMIC-IV to refer to the data set and MIMIC to refer to the multiple indicators multiple causes model.  Details of this model are provided in Supporting~Information~A.

For comparison, we used the same comparison methods described in the prior section, in addition to an unadjusted estimate.  Since the proxies are not interchangeable, we tested every possible even split of the proxies into NCOs and NCEs for PCI.  Confidence intervals (CIs) for all estimates were computed using 5000 bootstrap replications.  

A direct comparison of patients admitted to the hospital versus those discharged home, without adjusting for other covariates, estimated the ATE as 0.113 [0.097,  0.130], suggesting that hospital admission increases the risk of 30-day readmission for the average patient. IPW produced a much larger and more variable, estimate of 0.648 [-0.336,  1.637]. A basic linear model provided a smaller estimate of 0.074 [0.052,  0.095].  PCI resulted in a range of estimates from 0.044 to 0.172, with a median value of 0.089. The latent variable method yielded a distinctly different result: a numerically negative estimate of -0.021 [-0.192,  0.109]. This estimate suggests that hospitalization may reduce the risk of readmission for the average patient, in contrast to the positive estimates from the other methods. 

\section{Conclusion}\label{sec:discussion}
We presented an alternate method for using proxy variables in causal inference with unmeasured confounding. We use a factor analysis model to relate proxies and treatment to latent confounders, building covariates which feed into the outcome model that adjusts for unmeasured confounding. Our simulations demonstrated several key insights. First, categorizing proxies as confounders leads to biased estimates, underscoring the need for dedicated proxy methods. Our approach reduces this bias and achieves higher overall accuracy compared to standard linear regression and IPW under correct model specifications.

Moreover, we found our method to be surprisingly accurate under various forms of model misspecification. When using a linear model to approximate a quadratic relationship, assuming normality when latent variables are skew-normal, or assuming continuous treatments despite binary treatment generation, our estimates consistently produced less bias compared to most estimates produced by comparator methods. Our method was only outperformed when a direct confounder or an IV was incorrectly treated as a proxy, but even then, only by non-proxy based methods. In these cases our method became slightly biased, although estimates retained the correct sign and plausible magnitude.

In our case study involving older adults presenting to the ED for chest pain, our method uniquely estimated a numerically negative ATE of hospitalization on readmission, aligning better with prior literature and clinical expectations.\cite{avendano2020average} By contrast, comparison methods produced uniformly positive ATE estimates. This difference illustrates how carefully designed proxy variable methods can critically impact clinical interpretations and subsequent recommendations.

We acknowledge several limitations. First, our method cannot handle every possible relationship between the confounders $(X,U)$ and the outcome $Y$. This limits the expressiveness of our technique, though we did demonstrate that with $g(U,X)=[1, U]'$, we were able to capture accurate ATE estimates even when $Y$ had a quadratic dependence on $U$. Second, we computed confidence intervals using bootstrapping. This led to overly conservative estimates.  Additionally, our method may be biased when we treat a confounder as a proxy, as shown in simulation.  While we believe our ATE estimate to be more realistic in the case study, we would need an randomized controlled trial to validate our estimates.  A final limitation is that successful applications of our method depends not only on statistical performance but also on practitioner uptake. Whether practitioners eventually adopt the method will depend on how our method is received by the community.

Future work will aim to extend robustness and flexibility, incorporating a broader class of approximations for both the latent variable and its effect on the outcome and exploring analytical methods for confidence interval estimation. A promising direction includes developing automated approaches that accurately distinguish proxies and confounders without explicit user categorization. Ultimately, by providing another strategy for leveraging proxies for unmeasured confounding, we hope to encourage broader adoption of proxy-based causal inference methods.

\subsection*{Acknowledgements}
Research of Haley Colgate Kottler was supported in part by National Science Foundation Award DMS-2023239 through the Institute for Foundations of Data Science and Haley Colgate Kottler and Amy Cochran were supported in part by the University of Wisconsin Fall Competition.

\bibliographystyle{vancouver}
\bibliography{biblio.bib}

\newpage
\appendix
\section{Identification Proof}\label{ap:idproofs}

\noindent 
Proof of Theorem 1:
\begin{proof}
Suppose the relevant assumptions hold. For $a\in\A$, we have
    \begin{align}
        \E[Y|Z,A=a,X] &= \E[\E[Y|U,Z,A=a,X]|Z,A=a,X]\\
        &= \E[\E[Y(a)|U,Z,A=a,X]|Z,A=a,X]\\
        &= \E[\E[Y(a)|U,X]|Z,A=a,X]
    \end{align}
using the law of iterated expectation, consistency, and conditional independence. 

By Assumption~4.2, there exist $\tau(a,X)$ so that $\E[Y(a)|U,X] = \tau(a,X) g(U,X)$ for each $a \in \A$. Therefore,
\begin{align*}
    \E[Y|Z,A=a,X]
        = \tau(a,X)\E[g(U,X)|Z,A=a,X].
\end{align*}
Note that the left hand side is a function of the observed distribution, whereas 

\noindent
$\E[g(U,X)|Z,A=a,X]$ can be identified by the observed distribution by Assumption~4.1. Provided $\tau(a,X)$ is the unique solution to the above equation, then we can identify $\tau(a,X)$ from the observed distribution.
Then, consider 
\begin{align*}
    \E[Y(a)|Z,X] &= \E[\E[Y(a)|U,Z,X]|Z,X]\\
    &= \E[\E[Y(a)|U,X]|Z,X]\\
    &= \E[\tau(a,X)g(U,X)|Z,X]\\
    &= \tau(a,X)\E[g(U,X)|Z,X].
\end{align*} 
The first step comes from the law of iterated expectation, the second from conditional independence, the third by Assumption~4.2, and finally by linearity of expectations. Again by linearity, $$\mathbb{E}[Y(a)-Y(a')|Z=z,X]=(\tau(a,X)-\tau(a',X))\mathbb{E}[g(U)|Z=z,X].$$
It remains to argue that $\mathbb{E}[g(U,X)|Z=z,X]$ can be identified from the observed distribution. This comes from Assumption~4.1. Since $\E[g(U,X)|Z,A,X]$ can be identified from the observed distribution for $(Z,A,X)$, then $\E[g(U,X)|Z=z,X] = \E[\E[g(U,X)|Z=z,A,X]|Z=z,X]$ can be identified from the observed distribution as well.
\end{proof}

\section*{Supporting Information A: Multiple Indicators Multiple Causes Model Properties}\label{ap:mimic}
\begin{observation}
    From the definitions of the MIMIC model, it must be true that $$Z|U \sim MVN(\Lambda U+\nu, \Psi)$$ and $$Z\sim MVN(\nu, \Lambda\Lambda'+\Psi).$$  We also have 
\begin{align*}
    \E[U|X] = \E[\Gamma X + \varepsilon_U|X] = \Gamma X
\end{align*}
and 
\begin{align*}
    \E[(U-\Gamma X)(U-\Gamma X)'|X] = \E[\varepsilon_U\varepsilon_U'] = I.
\end{align*}
\end{observation}

For the moment, we reduce MIMIC to the usual factor analysis model by ignoring $X$.
\begin{observation}\label{matrixaug}
    Let $A=b'U+c+\varepsilon_A$. If we augment the necessary matrices, we can write $$\begin{bmatrix}
    Z\\A
\end{bmatrix} = \begin{bmatrix}
    \Lambda\\
    b'
\end{bmatrix}U + \begin{bmatrix}
    \nu\\c
\end{bmatrix} + \begin{bmatrix}
    \Psi^{1/2}\varepsilon_Z\\\varepsilon_A
\end{bmatrix}.$$
    Therefore it must be true that 

    $$\begin{bmatrix}
        Z\\A
    \end{bmatrix}|U\sim MVN\left(\begin{bmatrix}
    \Lambda\\
    b'
\end{bmatrix}U + \begin{bmatrix}
    \nu\\c
\end{bmatrix}, \begin{bmatrix}
        \Psi & \overset{\rightharpoonup}{0}\\
        \overset{\rightharpoonup}{0}' & 1
    \end{bmatrix}\right)$$
    
    and
    $$\begin{bmatrix}
        Z\\A
    \end{bmatrix}\sim MVN\left(\begin{bmatrix}
        \nu\\c
    \end{bmatrix},
    \begin{bmatrix}
        \Lambda\\b'
    \end{bmatrix} 
    \begin{bmatrix}
        \Lambda'&b
    \end{bmatrix}  +    
    \begin{bmatrix}
        \Psi & \overset{\rightharpoonup}{0}\\
        \overset{\rightharpoonup}{0}' & 1
    \end{bmatrix}\right)$$ where $\overset{\rightharpoonup}{0}$ is a $p\times1$ vector of zeros.      
    
    We will denote the covariance matrix $$\begin{bmatrix}
        \Psi & \overset{\rightharpoonup}{0}\\
        \overset{\rightharpoonup}{0}' & 1
    \end{bmatrix}$$ by $\Psi^*$.
\end{observation}



\begin{lemma}\label{Minvertible}
    Define $M:=(\Lambda'\Psi^{-1}\Lambda+I).$ If $\Lambda$ is full rank, $M$ is invertible.
\end{lemma}
\begin{proof}
    We will show $\Lambda'\Psi^{-1}\Lambda$ is positive definite since the sum of positive definite matrices is also positive definite, and therefore invertible.  Clearly $\Lambda'\Psi^{-1}\Lambda$ is symmetric.
    As an invertible covariance matrix, $\Psi$ must be positive definite. The inverse must also be positive definite since the inverse of the eigenvalues of a matrix are the eigenvalues of the matrix inverse. Since $\Lambda$ is full rank, $\Lambda v=0$ if and only if $v$ is the zero vector. For a nonzero vector $v\in\mathbb{R}^k$,
    \begin{align*}
        v'\Lambda'\Psi^{-1}\Lambda v &= (\Lambda v)'\Psi^{-1}(\Lambda v) > 0
    \end{align*} since $\Lambda v\neq0$ and $\Psi^{-1}$ is positive definite.
\end{proof}

\begin{lemma}\label{ufromz}
    Let $U\sim MVN(0,I)\in\mathbb{R}^k$, and $Z|U\sim MVN(\Lambda U+\nu,\Psi)\in\mathbb{R}^p$. Fix $z\in \mathbb{R}^p$ and let $d:=\Lambda'\Psi^{-1}(z-\nu)$. If $M:=(\Lambda'\Psi^{-1}\Lambda+I)$ is invertible, then $$U|Z=z\sim MVN(M^{-1}d, M^{-1}).$$ 
\end{lemma}
\begin{proof}
    Since $Z$ and $U$ have densities, we can consider 
    $f_{U|Z=z}(u) = f_{Z|U=u}(z)f_U(u)/f_Z(z)$ using Bayes rule. 
    Then, up to a constant in $z$,
    \begin{align*}
        \frac{f_{Z|U=u}(z)f_U(u)}{f_Z(z)} &\propto \exp\left(-\frac{1}{2}\left(z-\Lambda u-\nu\right)'\Psi^{-1}\left(z-\Lambda u-\nu)\right)\right)\exp\left(-\frac{1}{2}u'u\right)\\
        &= \exp\left(-\frac{1}{2}\left((z-\nu)'\Psi^{-1}(z-\nu)+u'(\Lambda'\Psi^{-1}\Lambda+I)u-2(z-\nu)'\Psi^{-1}\Lambda u\right)\right)\\
        &\propto \exp\left(-\frac{1}{2}\left(u - M^{-1} d\right)'M\left(u-M^{-1}d\right)\right) 
    \end{align*}
\end{proof}

Using the matrix augmentation from Observation \ref{matrixaug}, we arrive at the following corollary.

\begin{corr}\label{ufromzastar}
    Let $A=b'U+c+\varepsilon_A$ where $U,b,c,\varepsilon_A,$ and $Z$ are defined as in Assumption 4.3. If $M^*=\begin{bmatrix}
        \Lambda'&b
    \end{bmatrix}
   (\Psi^*)^{-1} 
    \begin{bmatrix}
        \Lambda\\b'
    \end{bmatrix}
    +I
    $ is invertible then $U|Z=z,A=a\sim MVN((M^*)^{-1}d^*,(M^*)^{-1})$ where $d^*=\begin{bmatrix}
        \Lambda' & b
    \end{bmatrix}
    (\Psi^*)^{-1}\begin{bmatrix}
        z-\nu\\a-c
    \end{bmatrix}
    $.
\end{corr}

We can compute $U|Z,X$ for the MIMIC model in a similar manner to the computation for $U|Z$ above.  We have
\begin{align*}
    f_{U|Z,X}(u) &= \frac{f_{U,Z,X}(u, z, x)}{f_{Z,X}(z, x)}\\ &= \frac{f_{Z|U}(z)f_{U|X}(u)f_X(x)}{f_{Z|X}(z)f_X(x)} \\&= \frac{f_{Z|U}(z)f_{U|X}(u)}{f_{Z|X}(z)} \\&\propto f_{Z|U}(z)f_{U|X}(u).
\end{align*}

Therefore
\begin{align*}
    f_{U|Z,X}(u) &\propto \exp\left(-\frac{1}{2}\left[(z-\Lambda U - \nu)'\Psi^{-1}(z-\Lambda U - \nu)' + (U - \Gamma X)'(U-\Gamma X)\right]\right) \\
    &\propto \exp\left(-\frac{1}{2}\left[U'\Lambda'\Psi^{-1}\Lambda U - 2 (z-\nu)'\Psi^{-1}\Lambda U + U'U -2X'\Gamma'U\right]\right)\\
    &\propto \exp\left(-\frac{1}{2}\left[U'(\Lambda'\Psi^{-1}\Lambda + I)U -2\left((z-\nu)'\Psi^{-1}\Lambda + X'\Gamma'\right)U\right]\right)\\
    &\propto \exp\left(-\frac{1}{2}(U'MU -2(d' + X'\Gamma')U\right)\\
    &\propto \exp\left(-\frac{1}{2}\left(U - M^{-1}(d+\Gamma X)\right)'M\left(U - M^{-1}(d+\Gamma X)\right)\right)
\end{align*}

so $U|Z,X\sim MVN(M^{-1}(d + \Gamma X), M^{-1})$.

Similarly to the process for incorporating $A$ previously, if we augment our matrices we can show that $U|Z,X,A\sim MVN((M^*)^{-1}(d^*+\Gamma X), (M^*)^{-1}).$

\section*{Supporting Information B: Simulated Data Robustness Checks}\label{app:robust}
Implementations are available at 

\url{https://github.com/HaleyColgateKottler/LatentProxyVars}.

\subsection*{Comparison methods}\label{app:compmethods}
We provide R code used for the comparison methods to enable replication.  Each function requires a dataframe of observations of $Z,A,Y$.

\vspace{.25in}
\noindent
We use a modification of the usual IPW algorithm to allow for continuous treatments.  Note that for binary treatments we use the usual IPW algorithm with logistic regression.
\begin{verbatim}
continuousIPWest <- function(df) {
  AZ <- subset(df, select = -c(Y))
  Z <- subset(AZ, select = -c(A))
  
  m1 <- glm(formula = A ~ ., family = gaussian(link = "identity"), data = AZ)
  meanAZ <- predict(m1, newdata = Z, type = "response")
  varAZ <- var(df$A - meanAZ)
  weightings <- dnorm(df$A, mean = meanAZ, sd = sqrt(varAZ))
  
  m2 <- glm(formula = A ~ 1, family = gaussian(link = "identity"), data = AZ)
  meanA <- predict(m2, newdata = df, type = "response")
  varA <- var(df$A - meanA)
  numerators <- dnorm(df$A, mean = meanA, sd = sqrt(varA))
  
  df$full.weights <- numerators / weightings
  df$A2 <- df$A^2
  design <- svydesign(
    id = ~1,
    weights = ~full.weights,
    data = df
  )
  
  m3 <- svyglm(
    formula = formula(paste("Y ~ A + A2 + ",
    paste0(colnames(Z), collapse = " +"))),
    family = gaussian(link = "identity"),
    data = df, design = design)
  
  newdata1 <- data.frame("A" = 1, Z, "A2" = 1)
  EY1 <- mean(predict(m3, newdata1))
  newdata0 <- data.frame("A" = 0, Z, "A2" = 0)
  EY0 <- mean(predict(m3, newdata0))
  
  return(EY1 - EY0)
}
\end{verbatim}

\vspace{.25in}
\noindent
We use a simple linear model for $Y$ in terms of $A$ and $Z$, and extract the coefficient on $A$ for the ATE estimate.
\begin{verbatim}
linearEst <- function(df) {
  m1 <- lm(Y ~ ., df)
  linear.ATE <- m1$coefficients["A"]

  return(linear.ATE)
}
\end{verbatim}

\vspace{.25in}
\noindent 
For proximal causal inference, you must additionally provide lists of the NCOs and NCEs that are included in the dataframe.

\begin{verbatim}
proximal_causal <- function(df, nco.names, nce.names){
  wformulas <- paste(nco.names,
                    " ~ A + ", paste0(nce.names, sep = "",
                    collapse = " + "),
                    " + ", paste0("A * ", nce.names, sep = "",
                    collapse = " + "), sep = "")
  m.hW <- lapply(wformulas, lm, data = df)
  wav <- sapply(m.hW, predict, data = df[,c("A", "Y", nce.names)])
  
  prior.names <- colnames(df)
  df <- cbind(df, wav)
  colnames(df) <- c(prior.names,
                    paste("wav", 1:length(nco.names), sep = ""))
  
  # estimate tau_a
  m1formula <- paste("Y ~ A + ",
                     paste0("wav", 1:length(nco.names), collapse = " + "),
                     " + ",
                     paste0("A * ", paste("wav", 1:length(nco.names),
                     sep = ""), collapse = " + "),
                     sep = "")
  m1 <- lm(m1formula, df)
  alpha <- m1$coefficients[c("(Intercept)",
                           paste("wav", 1:length(nco.names), sep = ""))]
  gamma <- m1$coefficients[c("A", paste("A:wav", 1:length(nco.names),
                           sep = ""))]
  
  # calc CATE
  mWVformulas <- paste(nco.names,
                      " ~ ", paste0(nce.names, collapse = " + "))
  m.WVs <- lapply(mWVformulas, lm, data = df)
  WVs <- lapply(m.WVs, predict, data = df)
  CATE <- gamma[1] + gamma[2:(1 + length(nco.names))] %*% 
            t(sapply(WVs, unlist))
  return(mean(CATE))
}

\end{verbatim}

\vspace{.25in}
\noindent
We implement a standard IV algorithm.
\begin{verbatim}
IVest <- function(df) {
  AZ <- subset(df, select = -c(Y))

  m1 <- lm(formula = A ~ ., data = AZ)
  df$x <- predict(m1, newdata = df)

  m2 <- lm(formula = Y ~ x, data = df)

  IV_ATE <- m2$coefficients[["x"]]
  return(IV_ATE)
}
\end{verbatim}

\subsection*{Parameters}
These are the parameters used for the synthetic data robustness study.

\subsubsection*{Consistency} Following the model introduced in Section~4.2, data is generated using 
$k = 2$, $p = 8$, $\nu=\overset{\rightharpoonup}{0}$, $$\lambda=\begin{bmatrix}
  0.534& -0.644\\
  0.000&  0.837\\
  0.425&  0.648\\
  0.207&  0.811\\
 -0.103&  0.889\\
 -0.020&  0.774\\
  0.256&  0.797\\
  0.894& -0.015
\end{bmatrix}$$

$diag(\psi) = [
    0.548, 0.548, 0.632, 0.548, 0.447, 0.632, 0.548, 0.447]$

    $b = [-0.113,  0.829]$, 
    $c = 0.2$

    $Y = \left(\begin{bmatrix}
        -0.8\\  0.5\\ -0.9
    \end{bmatrix} + A \begin{bmatrix}
        0.3\\ 0.9\\ 0.4
    \end{bmatrix}\right)\begin{bmatrix}1\\U\end{bmatrix} + \varepsilon_Y $
    where $\varepsilon_Y\sim\mathcal{N}(0,0.1)$

\subsubsection*{IV}
For this test, $U\sim\mathcal{N}(0,I)\in\R^2$ and $Z\sim\mathcal{N}(0,I)\in\R^6$.  The treatment assignment is generated via $A = \Lambda\begin{bmatrix}
    Z\\U
\end{bmatrix} + \varepsilon_A$ where $\varepsilon_A\sim\mathcal{N}(0,1)$.  Then $$Y=(\alpha+A\gamma)\begin{bmatrix}
    1\\U
\end{bmatrix} + \varepsilon_Y$$ where $\varepsilon_Y\sim\mathcal{N}(0,0.1)$.  Additionally, $
    \Lambda = [ 0.245, 0.614, 0.272, 0.140, 0.002, -0.007, -0.013, 0.263
],
$ $\alpha=[-0.6,0.2,0.7]$, and $\gamma=[0.3,0.8,0.4]$.

\subsubsection*{No latent variable}
For this test, $Z\sim\mathcal{N}(0,\Psi)\in\R^6$, where 

\begin{align*}
    diag(\Psi) = [&0.548, 0.632, 0.548, 0.447,0.632, 0.4472136, 0.632, 0.447].
\end{align*}
 Then $A = \Lambda Z + \varepsilon_A$ where $\varepsilon_A\sim\mathcal{N}(0,1)$ and 

\begin{align*}
\Lambda=[0.002, 0.692, 0.004, 0.347, -0.014, 0.005, 0.000, 0.008].    
\end{align*}
Finally, $Y=(\alpha+A\gamma)\begin{bmatrix}
    1\\Z
\end{bmatrix}+\varepsilon_Y$ with $\varepsilon_Y\sim\mathcal{N}(0,0.01)$, 

$\alpha=[-0.9,  0.7,  0.6, -0.6,  0.8,  0.3,  0.2, -0.5,  0.4]$ and 

$\gamma=[0.3,  0.4,  0.7,  0.5,  0.2, -0.8, -0.7,  0.9, -0.6]$.

\subsubsection*{Skew normal} This test followed the model from Section~4.2 except that $U\in\R$ is skew normal with shape parameter 5, location parameter $-1.26$, and scale parameter $1.606$.  Additionally, 

$\Lambda=[0.632, 0.548, -0.548, 0.632, 0.548]$, 

$diag(\Psi)=[0.775, 0.837, 0.837, 0.775, 0.837]$, 

$b=-0.632$, $c=0.2$, and $Y=([-0.9,0.7] + A[0.3, 0.6])\begin{bmatrix}
    1\\U
\end{bmatrix} + \varepsilon_Y$ where $\varepsilon_Y\sim\mathcal{N}(0,0.001)$.

\subsubsection*{Binary $A$}
For $U$ and $Z$, this followed the model from Section~4.2. The treatment was generated with a probit model via $$A=\begin{cases}
    1 & \text{ if }b'U+c + \varepsilon_A\\
    0 & \text{ otherwise}
\end{cases}$$ where $\varepsilon_A\sim\mathcal{N}(0,1)$.  Here $k=3$, $p=10$, 

$\Lambda=\begin{bmatrix}
  -0.033&  0.104& 0.830\\
  -0.027&  0.544& 0.062\\
  -0.042&  0.494& 0.674\\
   0.465& -0.041& 0.531\\
  -0.017&  0.332& 0.436\\
   0.142&  0.035& 0.528\\
   0.058&  0.169& 0.518\\
   0.256&  0.026& 0.796\\
   0.095&  0.196& 0.808\\
   0.038&  0.293& 0.783
\end{bmatrix},$

$
diag(\Psi)=[0.548, 0.837, 0.548, 0.707, 0.837, 0.837, 0.837, 0.548, 0.548, 0.548]$,

$b=[0.242, 0.039, 0.490]$, and $c=0.2$.  Finally, 

$Y=([0.7, 0.5, 0.6, -0.4]+A[0.3, 0.9, -0.8, 0.2])\begin{bmatrix}
    1\\U
\end{bmatrix}+\varepsilon_Y$ where $\varepsilon_Y\sim\mathcal{N}(0,0.01)$.

\subsubsection*{Ratio of number of proxies to dimension of latent variable}
For this test, data generation followed the model from Section~4.2 with $k=2$.  To allow for increasing numbers of proxies, parameters were chosen for $p=20$, and then truncated as necessary for smaller $p$ values.  We used

\begin{multline*}
    diag(\Psi) = [0.894, 0.894, 0.837, 0.775, 0.775, 0.837, \\0.837, 0.837, 0.894, 0.894, 0.894, 0.894, 0.894,\\ 0.775, 0.837, 0.837, 0.775, 0.894, 0.894, 0.894],
\end{multline*}

$b=[0.044, 0.631]$, $c=0$, 

$Y=([-0.7, -1.9, 1.3]+A[0.3, -0.2, 0.0])\begin{bmatrix}
    1\\U
\end{bmatrix} + \varepsilon_Y$ where $\varepsilon_Y\sim \mathcal{N}(0,0.1)$, and

$\Lambda=\begin{bmatrix}
   0.436&  0.098\\
   0.445&  0.048\\
   0.515&  0.188\\
  -0.054&  0.630\\
  -0.031&  0.632\\
   0.548& -0.009\\
   0.544&  0.064\\
   0.542&  0.079\\
  -0.038&  0.446\\
   0.435&  0.106\\
   0.435&  0.102\\
   0.191&  0.404\\
   0.447&  0.013\\
  -0.016&  0.632\\
   0.534&  0.120\\
  -0.037&  0.546\\
   0.290&  0.562\\
   0.144&  0.423\\
  -0.038&  0.446\\
   0.446& -0.035
\end{bmatrix}.$

\subsubsection*{Higher order outcome generation models}
Data generation for $(U,Z,A)$ followed Section~4.2 with $k=1$, $p=6$,

$\Lambda =\begin{bmatrix}
    -0.775\\ -0.894\\ 0.775\\ 0.775\\ -0.775\\ -0.775
\end{bmatrix},$

\begin{align*}
    diag(\Psi) = [0.632, 0.447, 0.632, 0.632, 0.632, 0.632],
\end{align*}

$b=0.837$ and $c=0$. Finally,

$Y = \left([0.5, 0.7, -0.8] + A[-0.4, -0.9, 0.7]\right)\begin{bmatrix}
    1\\U\\U^2
\end{bmatrix} + \varepsilon_Y$ with $\varepsilon_Y\sim\mathcal{N}(0,0.1)$.

\subsubsection*{Coverage} Following the model in Section~4.2, data was generated with $p=6$ and $k=2$ using

$\Lambda=\begin{bmatrix}
    0.498&  0.673\\
-0.099&  0.768\\
 0.886& -0.125\\
 0.600& -0.490\\
-0.617&  0.566\\
 0.326&  0.833
\end{bmatrix},$

$diag(\Psi)=[0.548, 0.632, 0.447, 0.632, 0.548, 0.447]$,

$b=[0.833, -0.078]$ and $c=0$. Finally,

$Y=([0.7, -0.5, -0.8]+A[0.5, 0.7, 0.2])\begin{bmatrix}
    1\\U
\end{bmatrix}+\varepsilon_Y$ where $\varepsilon_Y\sim\mathcal{N}(0,0.1)$.

\subsection*{Results with Comparison to the Instrument Variables Method}
Since instrument variable techniques are very sensitive to assumptions, the resulting estimates are often orders of magnitude different from other methods. However, for completeness we include those in Figure \ref{fig:lineplotpanel_withIV} and \ref{fig:boxplotpanel_withIV}.
\begin{figure}
    \centering
    \includegraphics[width=\linewidth]{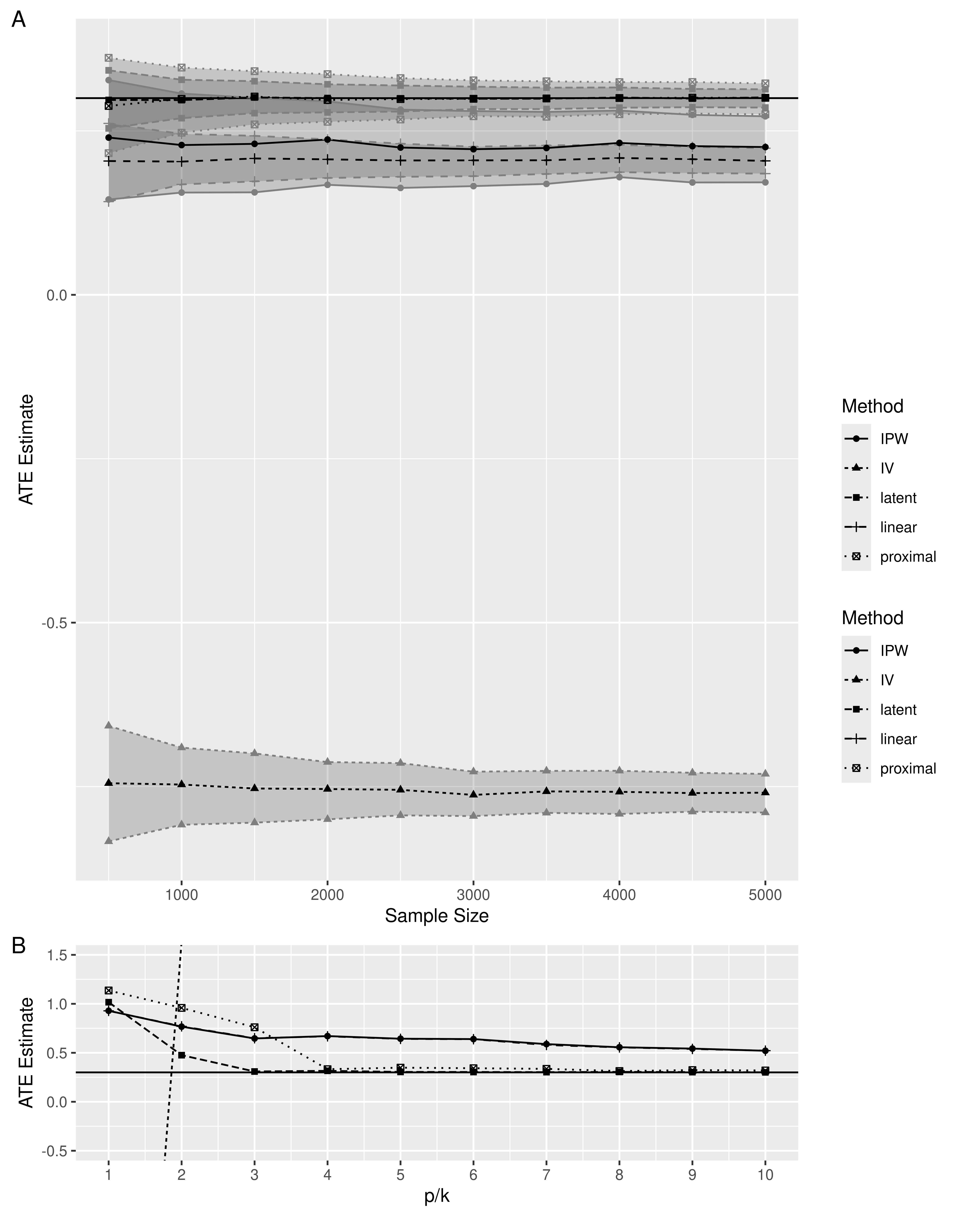}
    \caption{A) Estimation accuracy of our method and comparison methods as a function of sample size, reported as the median estimate and interquartile ranges (25th and 75th percentiles) across simulation replications. B) Estimation accuracy as function of different ratios of proxies to latent confounders ($p/k$), reported as the median estimate across simulation replications. Our method is the most accurate for all sample sizes, rivaled only by proximal causal inference which has a higher mean squared error.  With $p/k\geq2$ our method is more accurate than the comparison methods, though with higher $p/k$ ratios, proximal causal inference obtains similar accuracy.}
    \label{fig:lineplotpanel_withIV}
\end{figure}

\begin{figure}
    \centering
    \includegraphics[width=\linewidth]{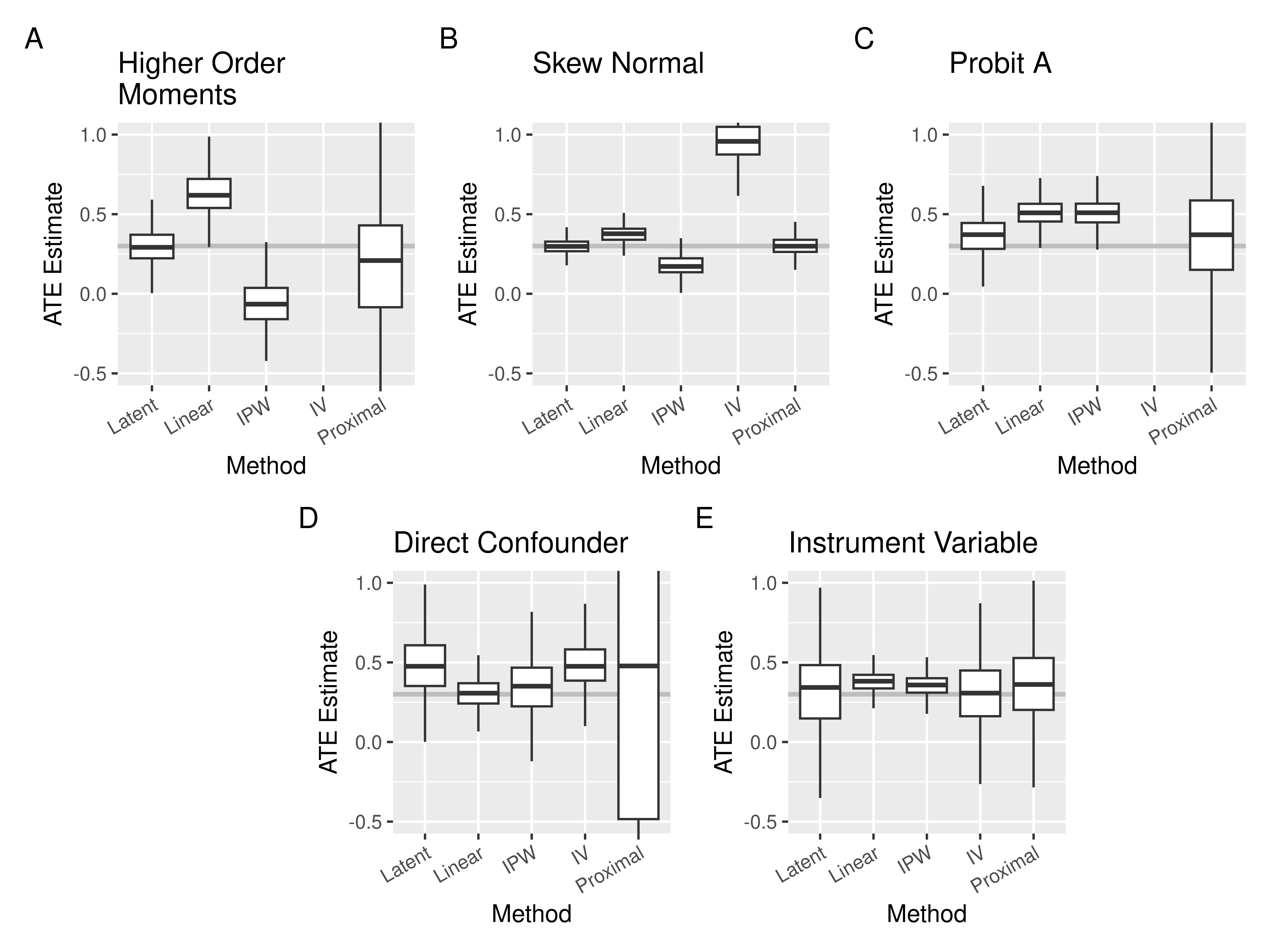}
    \caption{Performance of methods under different modeling and structural assumptions: A) the true outcome generation process is quadratic in $U$ but the estimator is linear, B) the latent variable is skew normal instead of normal, C) the treatment is binary but estimated as continuous, D) $Z$ is a confounder not a proxy, E) $Z$ is an instrument not a proxy. The centerline is the median, the hinges show the inter-quartile range (IQR), and the whiskers show the minimum of the highest estimate or 1.5IQR from the upper hinge and the maximum of the lowest estimate or 1.5IQR from the lower hinge. }
    \label{fig:boxplotpanel_withIV}
\end{figure}

\end{document}